\documentclass{sig-alternate-2013}

\newfont{\mycrnotice}{ptmr8t at 7pt}
\newfont{\myconfname}{ptmri8t at 7pt}
\let\crnotice\mycrnotice%
\let\confname\myconfname%

\usepackage[pass]{geometry}
\usepackage[latin1]{inputenc}

\mathcode`l="8000
\begingroup
\makeatletter
\lccode`\~=`\l
\DeclareMathSymbol{\lsb@l}{\mathalpha}{letters}{`l}
\lowercase{\gdef~{\ifnum\the\mathgroup=\m@ne \ell \else \lsb@l \fi}}%
\endgroup

\usepackage{color}
\usepackage[plainpages=false,pdfpagelabels,colorlinks=true,citecolor=blue,hypertexnames=false]{hyperref}
\usepackage{enumerate}
\usepackage{comment}
\usepackage{mathrsfs}
\usepackage{bm}
\usepackage{stmaryrd}
\usepackage[all]{xypic}
\usepackage{array,multirow}

\DeclareBoldMathCommand{\bbeta}{\beta}

\begin{document}

\definecolor{bleu}{rgb}{0,0,0.8}
\definecolor{orange}{rgb}{0.8,0.4,0}
\def\todo#1{{\color{red} #1}}
\def\comment#1{{\color{orange} #1}}
\def\gathen#1{{#1}}
\def\putvd#1{{#1}}

\newtheorem{theo}{Theorem}[section]
\newtheorem{lem}[theo]{Lemma}
\newtheorem{prop}[theo]{Proposition}
\newtheorem{cor}[theo]{Corollary}
\newtheorem{quest}[theo]{Question}
\newtheorem{ex}[theo]{Example}
\newtheorem{rem}[theo]{Remark}
\newtheorem{deftn}[theo]{Definition}

\newcommand{\N}{\mathbb N}
\newcommand{\F}{\mathbb F}
\newcommand{\Z}{\mathbb Z}
\newcommand{\Q}{\mathbb Q}
\newcommand{\Fp}{\F_p}
\newcommand{\Fpm}{\F_{p^m}}
\newcommand{\Fq}{\F_q}
\newcommand{\Zq}{\Z_q}
\newcommand{\Chi}{\Xi}
\newcommand{\calA}{\mathcal A}

\newcommand{\ring}{\mathfrak A}
\newcommand{\ringp}{\mathfrak B}
\newcommand{\lc}{\text{\rm lc}}
\newcommand{\rgcd}{\text{\sc rgcd}}

\newcommand{\softO}{O\tilde{~}}
\newcommand{\M}{\text{\rm M}}
\newcommand{\Mdp}{\M_{\text{\rm dp}}}
\newcommand{\pC}{\text{\rm pC}}

\newcommand{\trp}{{}^{\text t}}
\newcommand{\Card}{\text{\rm Card}}
\newcommand{\coeff}{\mathrm{Coeff}}

\newcommand{\Sdp}{\ell[[t]]^{\mathrm{dp}}}
\newcommand{\Sdpl}{\ell[[t]]^{\mathrm{dp}}}

\newcommand{\cen}{\mathcal Z}
\newcommand{\norm}{\mathcal N}

\newcommand{\Sol}{\text{\rm Sol}}

%
\permission{Publication rights licensed to ACM. ACM acknowledges that this contribution was authored or co-authored by an employee, contractor or affiliate of a national government. As such, the Government retains a nonexclusive, royalty-free right to publish or reproduce this article, or to allow others to do so, for Government purposes only.}

\conferenceinfo{ISSAC'15,}{July 6--9, 2015, Bath, United Kingdom..\\
{\mycrnotice{Copyright is held by the owner/author(s). Publication rights
licensed to ACM.}}}
\copyrightetc{ACM \the\acmcopyr}
\crdata{978-1-4503-3435-8/15/07\ ...\$15.00.\\
http://dx.doi.org/10.1145/2755996.2756674}

\clubpenalty=10000
\widowpenalty = 10000


\title{A Fast Algorithm for Computing the p-Curvature}
%
%
%
%
%

\numberofauthors{3} 
\author{
%
%
\alignauthor Alin Bostan\\
\affaddr{Inria (France)}\\
  \email{\normalsize \textsf{alin.bostan@inria.fr}}
\alignauthor Xavier Caruso\\
  \affaddr{Universit\'e Rennes 1}\\
  \email{\normalsize \textsf{xavier.caruso@normalesup.org}}
\alignauthor \'Eric Schost\\
  \affaddr{Western University}\\
  \email{\normalsize \textsf{eschost@uwo.ca}}
}

\date\today

\maketitle

\begin{abstract}
We design an algorithm for computing the $p$-curvature of a
differential system in positive characteristic $p$. For a system of
dimension $r$ with coefficients of degree at most $d$, its complexity
is $\softO (p d r^\omega)$ operations in the ground field (where
$\omega$ denotes the exponent of matrix multiplication), whereas the
size of the output is about $p d r^2$. Our algorithm is then
quasi-optimal assuming that matrix multiplication is (\emph{i.e.}
$\omega = 2$). The main theoretical input we are using is the
existence of a well-suited ring of series with divided powers for
which an analogue of the Cauchy--Lipschitz Theorem holds.
\end{abstract}

\vspace{1mm}
 \noindent
 {\bf Categories and Subject Descriptors:} \\
\noindent I.1.2 [{\bf Computing Methodologies}]:{~} Symbolic and Algebraic
  Manipulation -- \emph{Algebraic Algorithms}

 \vspace{1mm}

\medskip


\section{Introduction}\label{sec:intro}

We study in this article algorithmic questions related to linear
differential systems in positive characteristic. Let~$k$ be an
arbitrary field of prime characteristic~$p$, and $A$ be an $r \times
r$ matrix with entries in the field $k(x)$ of rational functions
over~$k$. A simple-to-define, yet very important object attached to
the differential system $Y' = AY$ is its so-called
\emph{$p$-curvature}. It is the $p$-th iterate
$\partial_A^p$ of the map $\partial_A : k(x)^r \rightarrow k(x)^r$
that sends $v$ to $v' - A v$. It turns out that it is $k(x)$-linear.
It is moreover classical that
its matrix with respect to the canonical basis of $k(x)^r$ is equal to
the term $A_p$ of the recursive sequence $(A_i)_i$ defined by
\begin{equation}\label{deq:pcurv}
	A_1 = -A \quad \text{and} \quad A_{i+1} = A'_i - A \cdot A_i \quad \text{for} \quad i \geq 1.
\end{equation}
In all what follows, we will thus deliberately identify the
matrix~$A_p$ with the \emph{$p$-curvature of $Y' = AY$}.  The above
recurrence yields an algorithm for computing it, sometimes referred to
as {\em Katz's algorithm}.

The $p$-curvature is related to solutions; it measures to what extent 
the usual Cauchy--Lipschitz theorem applies in characteristic~$p$. More 
precisely, at an ordinary point, the system $Y'=AY$ admits a 
fundamental matrix of power series solutions in $k[[x]]$ if and only if the 
$p$-curvature $A_p$ vanishes. In this case, the system $Y'=AY$ even 
admits a fundamental matrix of solutions which are rational functions in $k(x)$. More 
generally, the dimension of the kernel of $A_p$ is equal to the 
dimension of the space of rational function solutions of $Y'=AY$.

The primary importance of the notion of $p$-curvature relies in its
occurrence in one of the versions of the celebrated Grothendieck--Katz
conjecture~\cite{Katz72,Katz82,ChambertLoir00,Tang14}. This
conjecture, first formulated by Alexandre Grothendieck in the late
1960s, is a local-global principle for linear differential systems,
which states that a linear differential system with rational function
coefficients over a function field admits a fundamental matrix of
algebraic solutions if and only if its $p$-curvatures vanish for
almost all primes~$p$.

In computer algebra, $p$-curvature has been introduced by van der
Put~\cite{vanDerPut95,vanDerPut96}, who popularized it as a tool for
factoring differential operators in
characteristic~$p$. Cluzeau~\cite{Cluzeau03} generalized the approach
to the decomposition of differential systems over $k(x)$. The
$p$-curvature has also been used by Cluzeau and van
Hoeij~\cite{ClvH04} as an algorithmic filter for computing exponential
solutions of differential operators in characteristic zero.

Computing efficiently the $p$-curvature is in itself a challenging
problem, especially for large values of $p$. Our initial motivation
for studying this question emerged from concrete applications, in
lattice path combinatorics~\cite{BoKa08b,BoKa08a} and in statistical
physics~\cite{BBHMWZ08}.
In this article, we address the question of the computation of $A_p$ in good
complexity, with respect to three parameters: the dimension $r$ of the system
$Y'=AY$, the maximum degree $d$ of the rational function entries of $A$, and
the characteristic~$p$ of the ground field. In terms of these quantities, the
arithmetical size of $A_p$ is generically proportional to $p d r^2$ if $r>1$.

\medskip\noindent{\bf Previous work.}
Cluzeau~\cite[Prop.~3.2]{Cluzeau03} observed that the direct algorithm
based on recurrence~\eqref{deq:pcurv} has complexity $\softO( p^2d
r^\omega)$, where $\omega$ is the matrix multiplication exponent and
the soft-O notation $\softO(\,)$ hides polylogarithmic
factors. Compared to the size of the $p$-curvature, this cost is good
with respect to~$r$ and $d$, but not to~$p$. The first subquadratic
algorithm in $p$, of complexity $\softO(p^{1 + \omega/3})$, was
designed in~\cite[\S6.3]{BoSc09}. In some special cases, additional
partial results were obtained in~\cite{BoSc09}, notably an algorithm
of quasi-linear cost {$\softO(p)$} for certain systems of order
$r=2$. However, the question of designing a general algorithm for
computing
$A_p$ with quasi-linear complexity in $p$~remained open. In a related,
but different direction, the article~\cite{BoCaSc14} proposed an
algorithm for computing the characteristic polynomial of the
$p$-curvature in time essentially linear in $\sqrt{p}$, without
computing $A_p$ itself.

\medskip\noindent{\bf Contribution.} We prove that the $p$-curvature $A_p$
can be computed in quasi-linear time with respect to~$p$. More precisely, our
main result (Theorem~\ref{th:pcurv}) states that $\softO\big(p d r^\omega)$
operations in $k$ are sufficient for this task. This complexity
result is quasi-optimal not only with respect to the main parameter~$p$, but
also to~$d$; with respect to the dimension~$r$, it is as optimal as matrix
multiplication. Moreover the algorithm we obtain is highly parallelizable 
by design.
The key tools underlying the proof of Theorem~\ref{th:pcurv}
are the notion of divided power rings in characteristic $p$, and a new formula
for the $p$-curvature (Propositions~\ref{prop:pcurv} and~\ref{prop:psiSAp2}) in
terms of divided power series. Crucial ingredients are the fact that a
Cauchy--Lipschitz theorem for differential systems holds over divided power
rings (Proposition~\ref{prop:CL}) and the fact that Newton iteration can be
used to efficiently compute (truncations of) fundamental matrices of divided
power solutions.

\medskip\noindent{\bf Structure of the paper.} In
Section~\ref{subsec:diffop}, we recall the main theoretical properties of the
basic objects used in this article. Section~\ref{subsec:dp} is devoted to the
existence and the computation of solutions of differential systems in divided
power rings. In Section \ref{sec:pcurv}, we move to the main objective of the
article, the computation of the $p$-curvature: after relating $A_p$ to the
framework of divided powers, we describe our main algorithm for $A_p$, of
complexity $\softO(pdr^\omega)$. We conclude in Section~\ref{sec:timings} by describing
the implementation of our algorithm and some benchmarks.

\medskip\noindent{\bf Complexity measures.} Throughout this article,
we estimate the cost of our algorithms by counting arithmetic
operations in the base ring or field at unit cost.  

We use standard complexity notations. The letter $\omega$ refers to a
feasible exponent for matrix multiplication (\emph{i.e.} there exists
an algorithm for multiplying $n \times n$ matrices over a ring $\ring$
with at most $O(n^\omega)$ operations in $\ring$); the best known
bound is $\omega <2.3729$ from~\cite{LeGall14}. The soft-O notation
$\softO( \cdot)$ indicates that polylogarithmic factors are omitted;
in particular, we will use the fact that many arithmetic operations on
univariate polynomials of degree $d$ can be done in $\softO(d)$
operations: addition, multiplication, Chinese remaindering,
\emph{etc}, the key to these results being fast polynomial
multiplication~\cite{ScSt71,Schoenhage77,CaKa91,HaHoLe14}.  A general
reference for these questions is~\cite{GaGe03}.


\section{Theoretical setting}
\label{subsec:diffop}

We introduce and briefly recall the main properties of the theoretical
objects we are going to use in this article.  All the material
presented in this section is classical; a general reference
is~\cite{PuSi03}.

\medskip

\noindent
{\bf Definitions and notations.}
Let $\ring$ be a commutative ring with unit. We recall that a 
\emph{derivation} on $\ring$ is an additive map $' : \ring \to \ring$, 
satisfying the Leibniz rule
$(fg)' = f' g + f g'$ for all $f, g \in \ring$.
The image $f'$ of $f$ under the derivation is called the 
\emph{derivative} of $f$. {}From now on, we assume that $\ring$ is
equipped with a derivation. A \emph{differential system} with
coefficients in $\ring$ is an equation of the form $Y' = AY$
where $A$ is a given $r \times r$ matrix with coefficients in $\ring$
(for a certain positive integer $r$), the unknown $Y$ is a column
vector of length $r$ and $Y'$ denotes the vector obtained from $Y$ by
taking the derivative component-wise. The integer $r$ is called the
\emph{dimension} of the system.  We recall briefly that a linear
differential equation:
\begin{equation}
\label{eq:diffeq}
a_r y^{(r)} + \cdots + a_1 y' + a_0 y = 0
\quad \text{(with $a_i \in \ring$)}
\end{equation}
can be viewed as a particular case of a differential system. Indeed, 
defining the companion matrix
\begin{equation}\label{eq:comp}
C = \left( \begin{matrix}
  & &&  -\frac{a_0}{a_r} \\
1 & &&   -\frac{a_1}{a_r} \\
  &\ddots &&  \vdots \\
  &     &1& -\frac{a_{r-1}}{a_r}
\end{matrix} \right)
\end{equation}
and $A=\trp{C}$, the solutions of the system 
$Y' = AY$ are exactly the vectors of the form $\trp{(y, y', \ldots, 
y^{(r-1)})}$ where $y$ is a solution of \eqref{eq:diffeq}.
In this correspondence, the order of the differential equation agrees
with the dimension of the associated differential system.

\medskip

\noindent
{\bf Differential modules.} 
A \emph{differential module} over $\ring$ is a pair $(M, \partial)$ 
where $M$ is an $\ring$-module and $\partial : M \to M$ is an additive 
map satisfying a Leibniz-like rule, which is:
\begin{equation}
\label{eq:Leibniz}
\forall f \in \ring, \, \forall x \in M, \quad
\partial (f x) = f' \cdot x + f \cdot \partial(x).
\end{equation}
There exists a canonical one-to-one correspondence between differential 
systems and differential modules $(M,\partial)$ for which $M = \ring^r$ 
for some $r$: to a differential system $Y' = AY$ of dimension $r$, we 
attach the differential module $(\ring^r, \partial_A)$ where $\partial_A 
: \ring^r \to \ring^r$ is the function mapping $X$ to $X' - AX$. Under 
this correspondence, the solutions of $Y' = AY$ are exactly  
vectors in the kernel of $\partial_A$.

To a differential equation as~\eqref{eq:diffeq}, one can associate
the {\em differential operator} 
$L = a_r \partial^r + a_{r-1} \partial^{r-1} + \cdots + a_1 \partial + a_0$;
it lies in the non-commutative ring $\ring \langle \partial\rangle$,
endowed with the usual addition of polynomials and a multiplication
ruled by the relation $\partial \cdot f = f \cdot \partial + f'$ for
all $f \in \ring$ (note that, as often in the literature, we are using
$\partial$ to denote either the structure map of a differential
module, and a non-commutative indeterminate). 

Then, if $a_r$ is a unit in $\ring$, one can further associate to $L$
the quotient $\ring\langle \partial\rangle /\ring\langle
\partial\rangle L\simeq \ring^r$. The differential structure inherited
from $\ring \langle \partial \rangle$ makes it a differential module
with structure map $X \in \ring^r \mapsto X'+ C X$, where $C$ is the
companion matrix defined above; in other words, this is the module
$(\ring\langle \partial\rangle /\ring\langle
\partial\rangle L, \partial_{-C})$, with the previous notation.

\medskip

\noindent
{\bf Scalar extension.}
Let $\ring$ and $\ringp$ be two rings equipped with derivations and
let $\varphi : \ring \to \ringp$ be a ring homomorphism commuting
with derivation. From a given differential system $Y' = AY$ with
coefficients in $\ring$, one can build a differential system over
$\ringp$ by applying $\varphi$: it is $Y' = \varphi(A) Y$, where
$\varphi(A)$ is the matrix obtained from $A$ by applying $\varphi$
entry-wise.

This operation admits an analogue at the level of differential modules: 
to a differential module $(M, \partial)$ over $\ring$, we attach the
differential module $(M_\ringp, \partial_\ringp)$ over $\ringp$ where
$M_\ringp = \ringp \otimes_{\varphi,\ring} M$ and $\partial_\ringp : 
M_\ringp \to M_\ringp$ is defined by:
$$\forall f \in \ringp, \, \forall x \in M, \quad
\partial_\ringp(f \otimes x) = f' \otimes x + f \otimes \partial(x).$$
It is easily seen that if $(M, \partial)$ is a differential module 
associated to the system $Y' = AY$ then $(M_\ringp, \partial_\ringp)$ 
is that associated to the system $Y' = \varphi(A) Y$.

\medskip

\noindent
{\bf The $p$-curvature.}  Let $k$ be any field of
characteristic~$p$. We assume here that $\ring$ is the field $k(x)$
--- consisting of rational functions over $k$ --- equipped with the
standard derivation. The \emph{$p$-curvature} of a differential module
$(M,\partial)$ over $k(x)$ is defined as the mapping $\partial^p : M
\to M$. It follows from the Leibniz rule \eqref{eq:Leibniz} and the
fact that the $p$-th derivative of any element of $k(x)$ vanishes that
the $p$-curvature is $k(x)$-linear.

This definition extends to differential systems as follows: the 
$p$-curvature of the system $Y' = AY$ is the $k(x)$-linear map 
$\partial_A^p : M_A \to M_A$ where $(M_A, \partial_A)$ is the 
corresponding differential module. One can check that the matrix of 
$\partial_A^p$ (in the canonical basis of $M_A$) is the $p$-th term of 
the recursive sequence $(A_i)$ defined in~\eqref{deq:pcurv}.

Considering again a differential operator $L$ and the associated
differential module $(\ring\langle \partial\rangle /\ring\langle
\partial\rangle L, \partial_{-C})$, for the associated
companion matrix $C$, we obtain the usual recurrence $A_1=C$ and
$A_{i+1} = A'_i + C \cdot A_i$. The $p$-curvature of $\ring\langle
\partial\rangle /\ring\langle \partial\rangle L$ will simply be called
the $p$-curvature of $L$.


\section{Series with divided powers}\label{subsec:dp}

In all this section, we let $\ell$ be a ring in which $p$ vanishes.
We recall the definition of the divided power ring over $\ell$, and
its main properties --- mainly, a Cauchy--Lipschitz theorem that will
allow us to compute solutions of differential systems. We show how
some approaches that are well-known for power series solutions carry
over without significant changes in this context. 
Most results in this section are not new; those from \S\ref{ssec:dp} and \S\ref{ssec:divpow}
are implicitly contained in \cite{Berthelot74,BeOg78}, while the theoretical basis of \S\ref{ssec:C-L} is similar to~\cite{KePr00}. 


\subsection{The ring $\Sdp$}\label{ssec:dp}

Let $\Sdp$ be the ring of formal series of the form:
\begin{equation}
\label{eq:expanddp}
f = a_0 + a_1 \gamma_1(t) + a_2 \gamma_2(t) + \cdots + a_i \gamma_i(t)
+ \cdots
\end{equation}
where the $a_i$'s are elements of $\ell$ and each $\gamma_i(t)$ is a
symbol which should be thought of as $\frac{t^i} {i!}$. The
multiplication on $\Sdp$ is defined by the rule $\gamma_i(t) \cdot
\gamma_j(t) = \binom {i+j}i \cdot \gamma_{i+j}(t)$.

\begin{rem}
The ring $\Sdp$ is not the PD-envelope in the sense of 
\cite{Berthelot74,BeOg78} of $\ell[[t]]$ with respect to the ideal $(t)$ 
but its completion for the topology defined by the divided powers 
ideals. Taking the completion is essential to have an analogue of 
the Cauchy--Lipschitz Theorem (\emph{cf} Proposition \ref{prop:CL}).
\end{rem}
Invertible elements of $\Sdp$ are easily described: they are exactly 
those for which the ``constant'' coefficient $a_0$ is invertible in 
$\ell$. The ring $\Sdp$ is moreover endowed with a derivation defined by 
$f' = \sum_{i=0}^\infty a_{i+1} \gamma_i(t)$ for $f= \sum_{i=0}^\infty 
a_i \gamma_i(t)$. It then makes sense to consider differential systems 
over $\Sdp$. A significant difference
with power series is the existence of an integral operator: it maps
$f$ as above to $\int f = \sum_{i=0}^\infty a_i \gamma_{i+1}(t)$ 
and satisfies $(\int f)'=f$ for all $f$.

\medskip

\noindent
{\bf Divided power ideals.}  For all positive integers $N$, we denote
by $\Sdp_{\geq N}$ the ideal of $\Sdp$ consisting of series of the
form $\sum_{i \geq N} a_i \gamma_i(t)$. The quotient $\Sdp/\Sdp_{\geq
  N}$ is a free $\ell$-module of rank $N$ and a basis of
it is $(1, \gamma_1(t), \ldots, \gamma_{N-1}(t))$. In particular, for
$N = 1$, the quotient $\Sdp/\Sdp_{\geq 1}$ is isomorphic to $\ell$: in
the sequel, we shall denote by $f(0) \in \ell$ the reduction of an
element $f \in \Sdp$ modulo $\Sdp_{\geq 1}$. On the writing
\eqref{eq:expanddp}, it is nothing but the constant coefficient $a_0$
in the expansion of $f$.  

We draw the reader's attention to the fact that $\Sdp_{\geq N}$ is
\emph{not} stable under derivation, so the quotients $\Sdp /
\Sdp_{\geq N}$ do \emph{not} inherit a derivation.

\medskip

\noindent
{\bf Relationship with $\ell[t]$.}
There exists a natural map $\varepsilon : \ell[t] \to \Sdp$ taking a 
polynomial $\sum_i a_i t^i$ to $\sum_{i=0}^{p-1} a_i i! \cdot 
\gamma_i(t)$. The latter sum stops at $i = p-1$ because $i!$ becomes 
divisible by $p$ after that. Clearly, the kernel of $\varepsilon$ is 
the principal ideal generated by $t^p$. Hence $\varepsilon$ factors
through $\ell[t]/t^p$ as follows:
\begin{equation}
\label{eq:defiota}
\ell[t] \stackrel{\text{pr}}{\longrightarrow}
\ell[t]/t^p \stackrel{\iota}{\longrightarrow} \Sdp
\end{equation}
where $\text{pr}$ is the canonical projection taking a polynomial
to its reduction modulo $t^p$. We observe moreover that the ideal
$t^p \ell[t]$ is stable under derivation and, consequently, that the
quotient ring $\ell[t]/t^p$ inherits a derivation. Furthermore, the 
two mappings in \eqref{eq:defiota} commute with the derivation.


\subsection{Computations with divided powers}\label{ssec:divpow}

It turns out that the
$\gamma_n(t)$'s can all be expressed in terms of only few of them,
resulting in a more flexible description of the ring $\Sdp$. To make
this precise, we set $t_i = \gamma_{p^i}(t)$ and first observe that
$t_i^n = n! \cdot \gamma_{n p^i}(t)$ for all $i$ and $n$;
this is proved by induction on $n$, using the equalities
$$\textstyle
t_i^{n+1} = n! \cdot \gamma_{n p^i}(t) \cdot \gamma_{p^i}(t) = n!
\cdot \binom {(n+1)p^i}{p^i} \gamma_{(n+1) p^i}(t),$$
since Lucas' Theorem shows that $\binom {(n+1)p^i}{p^i} \equiv n+1
\pmod p$. In particular $t_i^p = 0$ for all $i$.

\begin{prop}
\label{prop:relti}
Let $n$ be a positive integer and $n = \sum_{i=0}^s n_i p^i$ its
writing in basis $p$. Then:
$$\gamma_n(t) = 
\gamma_{n_0}(t) \cdot \gamma_{n_1 p}(t) \cdots \gamma_{n_s p^s}(t) =
\frac {t_0^{n_0}}{n_0!} \cdot
\frac {t_1^{n_1}}{n_1!} \cdots
\frac {t_s^{n_s}}{n_s!}.$$
\end{prop}

\begin{proof}
The first equality is proved by induction on $s$ using the fact that
if $n = a + bp$ with $0\leq a < p$, then $\gamma_a \gamma_{bp} =
\gamma_n$, since $\binom{a+bp}{a}\equiv 1 \pmod p$.  The second
equality then follows from the relations $t_i^{n_i} = n_i! \cdot
\gamma_{n_i p^i}(t)$. \end{proof}

A corollary of the above proposition is that elements of $\Sdp$ can be
alternatively described as infinite sums of monomials $a_{n_0, \ldots,
  n_s} \cdot t_0^{n_0} \cdot t_1^{n_1} \cdots t_s^{n_s}$ where the
$n_i$'s are integers in the range $[0,p)$ and the coefficient $a_{n_0,
    \ldots, n_s}$ lies in~$\ell$. The product in $\Sdp$ is then the
  usual product of series subject to the additional rules $t_i^p = 0$
  for all $i$.


More precisely, restricting ourselves to some given precision of the
form $N=n p^s$, we deduce from the above discussion the following
corollary.
\begin{cor}
\label{cor:quoSdp}
For $N=n p^s$, with $s \in \N$ and $n \in \{1, \ldots, p\}$, there is
a canonical isomorphism of $\ell$-algebras:
$$\Sdp/\Sdp_{\geq N} \, \simeq \,
\ell[t_0, \ldots, t_s]/(t_0^p, \ldots, t_{s-1}^p, t_s^n).$$
\end{cor}
For instance, if we take $s=0$ and $N=n$ in $\{1,\dots,p\}$, we obtain
the isomorphism $\Sdp/\Sdp_{\geq N} \simeq \ell[t]/t^N$.

In terms of complexity, the change of bases between left- and
right-hand sides can both be done in $\softO(N)$ operations in $\ell$:
all the factorials we need can be computed once and for all for
$O(\min(N, p))$ operations; then each monomial conversion takes
$O(s)=O(\log(N))$ operations, for a total of $O( N \log(N)) =
\softO(N)$.
 
The previous corollary is useful in order to devise a multiplication
algorithm for divided powers, since it reduces this question to
multivariate power series multiplication (addition takes linear time
in both bases).
To multiply in $\ell[t_0, \ldots, t_s]/(t_0^p, \ldots, t_{s-1}^p,
t_s^n)$, one can use a direct algorithm: multiply and discard unwanted
terms. Using for instance Kronecker's substitution
and FFT-based univariate arithmetic, we find that a 
multiplication in $\Sdp$ at precision $N$ (\emph{i.e.} modulo 
$\Sdp_{\geq N}$) can be performed with 
$\softO(2^{\log_p N} N)$ operations in $k$. A solution that leads 
to a cost $N^{1+\varepsilon}$ for any $\varepsilon > 0$ is 
in~\cite{Schost05}, but the former result will be sufficient.



\subsection{The Cauchy--Lipschitz Theorem}  \label{ssec:C-L}

A nice feature of the ring $\Sdp$ --- which does not hold for
$\ell[[t]]$ notably --- is the existence of an analogue of the
classical Cauchy--Lipschitz theorem. This property will have a
fundamental importance for the purpose of our paper; see for
instance~\cite[Proposition~4.2]{KePr00} for similar considerations.

\begin{prop}
\label{prop:CL}
Let $Y' = AY$ be a differential system of dimension $r$ with 
coefficients in $\Sdp$. For all initial data $V \in \ell^r$
(considered as a column vector) the following Cauchy problem
has a unique solution in $\Sdp$:
$$\left\{\!
\begin{array}{ll}
Y' = A\cdot Y \smallskip \\
Y(0) = V.
\end{array} \right.$$
\end{prop}

\begin{proof}
Let us write the expansions of $A$ and $Y$:
$$A = \sum_{i=0}^\infty A_i \gamma_i(t)
\quad \text{and} \quad
Y = \sum_{i=0}^\infty Y_i \gamma_i(t)$$
where the $A_i$'s and $Y_i$'s have coefficients in $\ell$.
The Cauchy problem translates to $Y_0 = V$ and 
$Y_{n+1} = \sum_{i=0}^n \binom n i \cdot A_i \cdot Y_{n-i}$.
It is now clear that it has a unique solution.
\end{proof}

Of course, Proposition \ref{prop:CL} extends readily to the case where 
the initial data $V$ is any matrix having $r$ rows. In particular, 
taking $V = I_r$ (the identity matrix of size $r$), we find that
there exists a unique $r \times r$ matrix $Y$ with coefficients in
$\Sdp$ such that $Y(0) = I_r$ and $Y' = A\cdot Y$. This matrix is often
called a \emph{fundamental system of solutions}.

\medskip 

\noindent
{\bf Finding solutions using Newton iteration.}  In characteristic
zero, it is possible to compute power series solutions of a
differential system such as $Y' = A\cdot Y$ using Newton iteration; an
algorithm for this is presented on \cite[Fig.~1]{BCOSSS07}.

One can use 
this algorithm to compute a
fundamental system of solutions in our context.  For this, we first
need to introduce two notations. Given an element $f \in \Sdp$ written
as $f = \sum_i a_i \gamma_i(t)$ together with an integer $m$, we set
$\lceil f \rceil^m = \sum_{i=0}^{m-1} a_i \gamma_i(t)$.
Similarly, if $M$ is a matrix with coefficients in $\Sdp$, we define 
$\lceil M \rceil^m$ and $\int M$ by applying the corresponding
operations entry-wise.

\noindent\hrulefill

\noindent {\bf Algorithm} {\tt fundamental\_solutions}

\noindent{\bf Input:} a differential system $Y' = AY$, an integer $N$

\noindent{\bf Output:} the fund. system of solutions modulo 
$\Sdp_{\geq N}$

\smallskip\noindent 1.\ $Y = I_r + t \: A(0)$;\, $Z = I_r$;\, $m = 2$

\smallskip\noindent 2.\ {\bf while} $m \leq N/2$:

\smallskip\noindent 3.\ \hspace{0.4cm} 
$Z = Z + \big\lceil Z (I_r - YZ) \big\rceil^m$

\smallskip\noindent 4.\ \hspace{0.4cm} 
$Y = Y - \Big\lceil Y \big( \int Z \cdot (Y' - \lceil A \rceil^{2m-1} Y) 
\big) \Big\rceil^{2m}$

\smallskip\noindent 5.\ \hspace{0.4cm} $m = 2m$

\smallskip\noindent 6.\ {\bf return} $Y$

\noindent\hrulefill

\smallskip

\noindent
Correction is proved as in the classical case
\cite[Lemma~1]{BCOSSS07}.

Let us take $n \in \{2,\dots,p\}$ and $s \in \N$ such that $n-1$ is
the last digit of $N$ written in basis $p$, and $s$ the corresponding
exponent; then, we have $(n-1)p^s \le N < n p^s$.  Since we are only
interested in costs up to logarithmic factors, we may assume that we
do all operations at precision $n p^s$ (a better analysis would take
into account the fact that the precision grows quadratically). 

By Corollary~\ref{cor:quoSdp} and the discussion that follows,
arithmetic operations in $\Sdp/\Sdp_{\geq n p^s}$ take time
$\softO(2^{\log_p N} N)$. This is also the case for differentiation
and integration, in view of the formulas given in the previous
subsection; truncation is free. The total complexity of Algorithm {\tt
  fundamental\_solutions} is therefore $\softO(2^{\log_p N} N
r^\omega)$ operations in $\ell$, where $r$ is the dimension of the
differential system.  If $N=p^{O(1)}$, which is what we need later on,
this is $\softO(N r^\omega)$.

\medskip

\noindent
{\bf The case of differential operators.}
We now consider the case of the differential system associated to a
differential operator $L = a_r \partial^r + \cdots + a_1 \partial
+ a_0 \in \Sdp\langle \partial \rangle$.
We will work under the
following few assumptions: we assume that $a_r$ is invertible, and
that there exists an integer $d < p$ such that all $a_i$'s can be
written
$a_i = \alpha_{i,0} + \alpha_{i,1} \gamma_1(t) + \cdots +
\alpha_{i,d} \gamma_d(t)$ for some coefficients $\alpha_{i,j}$ in
$\ell$; thus, by assumption, $\alpha_{r,0}$ is a unit in $\ell$.
Our goal is still to compute a basis of solutions up to precision $N$;
the algorithm is a direct adaptation of a classical construction to
the case of divided powers.

In all that follows, we let $f_0,\dots,f_{r-1}$ be
the solutions of $L$ in $\Sdp$, such that $f_i$ is the unique solution
of the Cauchy problem (\emph{cf} Proposition \ref{prop:CL}):
\begin{equation}
\label{eq:systop}
L(f_i) = 0 
\quad ; \quad
f_i^{(j)}(0) = \delta_{ij} \quad \text{for } 0 \leq j < r
\end{equation}
where $\delta_{ij}$ is the Kronecker delta. For $f =
\sum_{j=0}^\infty \xi_j \gamma_j(t)$ in $\Sdpl$, a direct computation
shows that the $n$-th coefficient of $L(f)$ is
$\sum_{i=0}^r \sum_{j=0}^n \binom n j \alpha_{i,j} \xi_{n+i-j}$.
Assume $L(f)=0$. Then, extracting the term in $\xi_{n+r}$, and 
using that $\alpha_{i,j}=0$ for $j > d$, we get
$\xi_{n+r} = \frac {-1} {\alpha_{r,0}} \sum_{i=0}^{r-1} \sum_{j=0}^d 
\binom n j \alpha_{i,j} \xi_{n+i-j}$. Letting $m=i{-}j$, we find
$\xi_{n+r} = \sum_{m=-d}^{r-1} A_m(n) \xi_{n+m}$ with 
$$A_m(n) =  \frac {-1} {\alpha_{r,0}}  \sum_{i=0}^{r-1} 
{\textstyle \binom n {i-m}} \alpha_{i,i-m}\\
 = \!\!\!\!
\sum_{\substack{0 \le i \le r-1 \vspace{0.5mm} \\ 0 \le i-m \le d}}
\frac {-\alpha_{i,i-m}}{\alpha_{r,0} (i-m)!}\, n^{\underline{i-m}}$$
and $n^{\underline{i-m}} = n(n-1) \cdots (n-(i-m-1))$ is a falling
factorial. The expression above for $A_m$ is well-defined, since we
assumed that $d < p$, and shows that $A_m$ is a polynomial of degree
at most $d$.

From this, writing the algorithm is easy. We need two subroutines:
${\tt from\_falling\_factorial}(F)$, which computes the expansion on
the monomial basis of a polynomial of the form $F=\sum_{0 \le j \le n}
f_j n^{\underline j}$, and ${\tt eval}(F, N)$, which computes the
values of a polynomial $F$ at the $N$ points $\{0,\dots,N{-}1\}$. The
former can be done using the divide-and-conquer algorithm
of~\cite[Section 3]{BoSc05} in time $\softO(n)$; the latter by the
algorithm of~\cite[Chapter~10]{GaGe03}, in time $\softO(\deg(F) +
N)$.
The previous discussion leads to the algorithm 
\texttt{solutions\_operator} below. In view of the previous discussion, 
the cost analysis is straightforward (at step 2., notice that all 
required factorials can be computed in time $O(d)$). The costs reported 
in the pseudo-code indicate the {\em
  total} amount of time spent at the corresponding line.

\noindent\hrulefill

\noindent {\bf Algorithm} {\tt solutions\_operator}

\noindent{\bf Input:} a differential operator $L \in \Sdpl\langle
\partial\rangle$ of bidegree $(d,r)$, with $d < p$; an integer $N$

\noindent{\bf Output:} the solutions $f_0,\dots,f_{r-1}$ at precision $N$

\smallskip\noindent 1.\ {\bf for} $m=-d,\dots,r-1$:

\smallskip\noindent 2.\ \hspace{0.4cm} $\hat A_m = \sum_{0 \le i \le r-1, 0 \le i-m \le d} \frac {-\alpha_{i,i-m}}{\alpha_{r,0} (i-m)!}\, x^{\underline{i-m}}$

\noindent \phantom{2.\ \hspace{0.4cm}}
{\sc Cost:} $O(d(r+d))$

\smallskip\noindent 3.\ \hspace{0.4cm} $A_m = {\tt from\_falling\_factorial}(\hat A_m)$

\noindent \phantom{3.\ \hspace{0.4cm}}
{\sc Cost:} $\softO(d(r+d))$

\smallskip\noindent 4.\ \hspace{0.4cm} Store ${\tt eval}(A_m, N-r)$

\noindent \phantom{4.\ \hspace{0.4cm}}
{\sc Cost:} $\softO((d+N)(r+d))$

\smallskip\noindent 5.\ {\bf for} $i=0,\dots,r-1$:

\smallskip\noindent 6.\ \hspace{0.4cm} $f_i=[0,\dots,0,1,0,\dots,0]$ ($i$th unit vector of length $r$)

\noindent \phantom{6.\ \hspace{0.4cm}}
{\sc Cost:} $O(r^2)$

\smallskip\noindent 7.\ \hspace{0.4cm} {\bf for} $n=0,\dots,N-r-1$:

\smallskip\noindent 8.\ \hspace{0.8cm} $f_{i,n+r} = \sum_{m=-d}^{r-1} A_m(n) f_{i,n+m}$

\noindent \phantom{8.\ \hspace{0.8cm}}
{\sc Cost:} $O(r N (r+d))$

\smallskip\noindent 9.\ {\bf return} $f_0,\dots,f_{r-1}$

\noindent\hrulefill

\smallskip

Altogether, we obtain the following result, where we use the
assumption $N > d$ to simplify slightly the cost estimate.

\begin{lem}
\label{lem:fm}
Suppose that $p < d$. Given a positive $N > d$, the classes of $f_0,
\ldots, f_{r-1}$ modulo $\Sdpl_{\geq N}$ can be computed with at most
$O(r N (r+d))$ operations in $\ell$.
\end{lem}

In particular, Algorithm {\tt solutions\_operator} has a better cost
than {\tt fundamental\_solutions} when $d = O(r^{\omega-1})$.


\section{Computing the p-curvature}\label{sec:pcurv}

In all this section, we work over a field $k$ of characteristic $p >
0$.  We consider a differential system $Y' = AY$ of dimension $r$ and
denote by $A_p$ the matrix of its $p$-curvature. We write $A = \frac
1{f_A} \tilde A$, where $f_A$ is in $k[x]$ and $\tilde A$ is a matrix
with polynomial entries. Let $d = \max(\deg f_A, \deg \tilde A)$, where
$\deg \tilde A$ is the maximal degree of the entries of $\tilde A$.  
We recall (\cite[Prop. 3.2]{Cluzeau03}, \cite[Lemma 1]{BoSc09}) a bound on the size of $A_p$. The bound follows from the recurrence~\eqref{deq:pcurv}, and it is tight.

\begin{lem}
\label{lem:bound}
The entries of the matrix $f_A^p {\cdot} A_p$ are all polynomials of 
degree at most $dp$.
\end{lem}

The goal of this section is to prove the following theorem.

\begin{theo}
\label{th:pcurv}
There exists an algorithm (presented below) which computes the matrix
of the $p$-curvature of the differential system $Y' = AY$ in
$\softO\big(p d r^\omega)$ operations in $k$.
\end{theo}
It is instructive to compare this cost with the size of the output. By
Lemma \ref{lem:bound}, the latter is an $r \times r$ matrix whose
entries are rational functions whose numerator and denominator have
degree $\simeq pd$, so its size is roughly $p d r^2$ elements of
$k$. Our result $\softO\big(p d r^\omega)$ is quasi-optimal if we
assume that matrix multiplication can be performed in quasi-optimal
time.


\subsection{A formula for the p-curvature}

Let $A_p$ denote the matrix of the $p$-curvature of the differential
system $Y' = AY$ (in the usual monomial basis). The expression of $A_p$
given at the very end of \S \ref{subsec:diffop} is unfortunately not
well-suited for fast computation.  The aim of this subsection is to
give an alternative formula for $A_p$ using the framework of divided
powers.

In order to relate $k(x)$ and a ring $\Sdp$, we pick a separable
polynomial $S \in k[x]$ which is coprime with $f_A$ and set $\ell =
k[x]/S$ (which is thus not necessarily a field). Let $a \in \ell$ be
the class of $x$. We consider the ring homomorphism:
$$\begin{array}{rcl} \varphi_S : \quad k[x] & \to &
  \ell[t]/t^p \smallskip \\ f(x) & \mapsto & f(t+a) \text{ mod }
  t^p. \end{array}$$ Regarding the differential structure, we observe
that $\varphi_S$ commutes with the derivation when $\ell[t]/t^p$ is
endowed with the standard derivation $\frac d{dt}$.  We furthermore
deduce from the fact that $S$ and $f_A$ are coprime that $\varphi_S$
extends to a homomorphism of differential rings $k[x][\frac 1{f_A}]
\to \ell[t]/t^p$ that we continue to denote by $\varphi_S$. We set
$\psi_S = \iota \circ \varphi_S$ where $\iota$ is the canonical
inclusion $\ell[t]/t^p \hookrightarrow \Sdp$ (\emph{cf} \S
\ref{subsec:dp}). As before, $\psi_S$ commutes with the derivation.
Finally, because $S$ is separable, we can check that $\varphi_S$ is
surjective and its kernel is the ideal generated by $S^p$. Hence
$\varphi_S$ induces an isomorphism:
\begin{equation}
\label{eq:isomquo}
\textstyle
k[x]/S^p = k[x][\frac 1{f_A}]/S^p
\stackrel{\sim}{\longrightarrow} \ell[t]/t^p.
\end{equation}

Let $Y_S$ be a fundamental system of solutions of the differential
system $Y' = \psi_S(A) \cdot Y$, \emph{i.e.} $Y_S$ is an $r \times r$
matrix with coefficients in $\Sdp$ such that $Y_S(0) = I_r$ and $Y_S'
= \psi_S(A) \cdot Y_S$. The existence of $Y_S$ is guaranteed by
Proposition \ref{prop:CL}. Moreover, the matrix $Y_S$ is invertible
because $Y_S(0) = I_r$ is.

\begin{prop}\label{prop:pcurv}
Keeping the above notations, we have:
\begin{equation}
\label{eq:psiSAp}
\varphi_S(A_p) = - Y_S^{(p)} \cdot Y_S^{-1}
\end{equation}
where $Y_S^{(p)}$ is the matrix obtained from $Y_S$ by taking the
$p$-th derivative entry-wise.
\end{prop}

\begin{proof}
We set $Z_S = Y_S^{-1}$ and let $(M, \partial)$ denote the differential
module over $\Sdp$ associated to the differential system $Y' = \psi_S(A)
Y$. Let $y_1, \ldots, y_r$ denote the column vectors of $Y_S$. They are
all solutions of the system $Y' = \psi_S(A) Y$, meaning that $\partial
(y_i) = 0$ for all $i$. Furthermore, if $(e_1, \ldots, e_r)$ is the
canonical basis of $(\Sdp)^r$, we have the matrix relations:
$\trp{Y_S} \cdot \underline e = \underline y$ and
$\underline e = \trp{Z_S} \cdot \underline y$
where $\underline y$ (resp. $\underline e$) is the column vector whose
coordinates are the \emph{vectors} $y_i$'s (resp. the $e_i$'s). 
Applying $\partial$ to the above relation, we find
$\partial (\underline e) = \trp{Z'_S} \cdot \underline y
+ \trp{Z_S} \cdot \partial(\underline y) =  \trp{Z'_S} \cdot \underline y$
and iterating this $p$ times, we deduce
$\partial^p (\underline e) = 
\trp{Z_S^{(p)}} \cdot \underline y = 
\trp{Z_S^{(p)}} \cdot \trp{Y_S} \cdot \underline e$.
On the other hand, the matrix $\psi_S(A_p)$ of the $p$-curvature is 
defined by the relation $\partial^p (\underline e) = \trp{\psi_S(A_p)} 
\cdot \underline e$. Therefore we get $\psi_S(A_p) = Y_S \cdot 
Z_S^{(p)}$.
Now differentiating $p$ times the relation $Y_S Z_S = I_r$, we find
$Y_S^{(p)} Z_S + Y_S \cdot Z_S^{(p)} = 0$. Combining this with the
above formula for $\psi_S(A_p)$ concludes the proof.
\end{proof}

In our setting, the matrix $A_p$ has coefficients in $k[x][\frac 1
  {f_A}]$ (\emph{cf} Lemma \ref{lem:bound}), from which we deduce that
$\psi_S(A_p)$ has actually coefficients in the subring $\ell[t]/t^p$ of
$\Sdp$. Therefore, using Eq.~\eqref{eq:psiSAp}, one can compute
$\psi_S(A_p)$ knowing only $Y_S$ modulo the ideal $\Sdp_{\geq
  2p}$. 

One can actually go further in this direction and establish a variant
of Eq.~\eqref{eq:psiSAp} giving an expression of $\psi_S(A_p)$ which
involves only the reduction of $Y_S$ modulo $\Sdp_{\geq p}$. To make
this precise, we need an extra notation. Given an integer $i \in
[0,p)$ and a polynomial $f \in \ell[t]/t^p$ (resp. a matrix $M$ with
  coefficients in $\ell[t]/t^p$), we write $\coeff(f,i)$
  (resp. $\coeff(M,i)$) for the coefficient in $t^i$ in $f$ (resp. in
  $M$).

\begin{prop}
\label{prop:psiSAp2}
Keeping the above notations, we have:
\begin{align}
\psi_S(A_p) & = - \bar Y_S \cdot Y_S^{(p)}(0) \cdot \bar Y_S^{-1} \nonumber \\
& = \bar Y_S \cdot \coeff(A \cdot \bar Y_S, \,p{-}1) \cdot \bar Y_S^{-1}
\label{eq:psiSAp2}
\end{align}
where we have set $\bar Y_S = Y_S \, \bmod \, \Sdp_{\geq p}$.
\end{prop}

\begin{proof}
Differentiating $p$ times the relation $Y_S' = \psi_S(A) \cdot Y_S$, we
observe that $Y_S^{(p)}$ is solution of the same differential system 
$Y' = \psi_S(A) Y$. Hence, thanks to uniqueness in Cauchy--Lipschitz
Theorem, we have the relation $Y_S^{(p)} = Y_S \cdot Y_S^{(p)}(0)$.
The first part of the Proposition follows by plugging this in 
Eq.~\eqref{eq:psiSAp} and reducing the result modulo $\Sdp_{\geq p}$. To establish the second part, it is now enough
to notice that the relation $Y_S' = \psi_S(A) \cdot Y_S$ implies:
$$Y_S^{(p)}(0) = (A \cdot Y_S)^{(p-1)}(0) 
= - \coeff(A \cdot \bar Y_S,\, p{-}1)$$
the minus sign coming from $(p-1)! \equiv -1 \pmod p$.
\end{proof}

\begin{rem}
We can rephrase Proposition \ref{prop:psiSAp2} as follows: letting
$y_1, \ldots, y_r$ denote the column vectors of $Y_S$ and $\bar y_i
\in (\ell[t]/t^p)^r$ be the reduction of $y_i$, the $p$-curvature of
$A$ modulo $t^p$ is the linear endomorphism of $(\ell[t]/t^p)^r$ whose
matrix in the basis $(\bar y_1, \ldots, \bar y_r)$ is $\coeff(A \cdot
\bar Y_S, \,p{-}1)$.  It is worth remarking that the latter matrix has
coefficients in the subring $\ell$ of $\ell[t]/t^p$.
\end{rem}

Remembering Eq.~\eqref{eq:isomquo}, we conclude that Proposition 
\ref{prop:psiSAp2} allows us to compute the image of the $p$-curvature 
$A_p$ modulo $S^p$. The strategy of our algorithm now becomes clear:
we first compute $A_p$ modulo $S^p$ for various polynomials $S$ and, when we
have collected enough congruences, we put them together to reconstruct 
$A_p$. The first step is detailed in \S \ref{subsec:local} just below
and the second step is the subject of \S \ref{subsec:glueing}.


\subsection{Local calculations}
\label{subsec:local}

In all this subsection, we fix a separable polynomial $S \in k[x]$
and denote by $m$ its degree. Our goal is to design an algorithm for
computing the matrix $A_p$ modulo $S^p$. 
After Proposition \ref{prop:psiSAp2}, the main remaining algorithmic
issue is the effective computation of the isomorphism $\varphi_S$ and
its inverse. 

\medskip

\noindent {\bf Applying $\varphi_S$ and its inverse.}
We remark that $\varphi_S$ factors as follows:
$$\begin{array}{ccccc} 
 k[x] / S^p  & \to  &  k[x,t]/\langle S, (t-x)^p\rangle & \to &  k[x,t]/\langle S, t^p\rangle \smallskip \\ 
 x & \mapsto & t & \mapsto & t+a. \end{array}$$ 
Applying the right-hand mapping, or its inverse, amounts to doing
a polynomial shift in degree $p$ with coefficients in $k[x]/S$. Using 
the divide-and-conquer algorithm of~\cite{GaGe97}, this
can be done in $\softO(p)$ arithmetic operations in $k[x]/S$,
which is $\softO(p m)$ operations in $k$. Thus, we are left with the 
left-hand factor, say $\varphi^\star_S$. Applying it is straightforward
and can be achieved in $\softO(pm)$ operations in $k$. It then only 
remains to explain how one can apply efficiently ${\varphi^\star_S}^{-1}$.

We start by determining the image of $x$ by ${\varphi^\star_S}^{-1}$;
call it $y={\varphi^\star_S}^{-1}(x)$; we may identify it with its
canonical preimage in $k[x]$, which has degree less than $pm$.
Write $y=\sum_{0 \le i < p} \zeta_i(x^p) x^i$, with every $\zeta_i$ in
$k[x]$ of degree less than $m$ (so that $\zeta_i(x^p)$ has degree less
than $pm$). Its image through $\varphi^\star_S$ is $\sum_{0 \le i < p}
\zeta_i(t^p) t^i$, which is $\sum_{0 \le i < p} \zeta_i(x^p) t^i$,
since $x^p=t^p$ in $k[x,t]/\langle S, (t-x)^p\rangle$.

Since $\varphi^\star_S(y)=x$, we deduce that $\zeta_0(x^p)=x \bmod S$
and $\zeta_i(x^p)=0 \bmod S$ for $i =1,\dots,p-1$. The first equality
implies that $x^p$ generates $k[x]/S$, so the fact that $\zeta_0$ has
degree less than $m$ implies that $\zeta_0$ is the unique polynomial
with this degree constraint such that $\zeta_0(x^p)=x \bmod S$. The other
equalities then imply that $\zeta_i=0$ for  $i=1,\dots,p-1$.

In order to compute $\zeta_0$, we first compute $\nu=x^p \bmod S$,
using $\softO(m \log(p))$ operations in $k$. Then, we have to find the
unique polynomial $\zeta_0$ of degree less than $m$ such that
$\zeta_0(\nu) = x \bmod S$. In general, one can compute $\zeta_0$ in
$O(m^\omega)$ operations in $k$ by solving a linear system. In the
common case where $m < p$, there exists a better solution.  Indeed,
denote by ${\rm tr}: k[x]/S \to k$ the $k$-linear trace form and write
$t_i={\rm tr}(\nu^i)$ and $t'_i={\rm tr}(x \nu^i)$, for
$i=0,\dots,m-1$. Then formulas such as those in~\cite{Rouillier99}
allow us to recover $\zeta_0$ from ${\bf t}=(t_0,\dots,t_{m-1})$ and
${\bf t'}=(t'_0,\dots,t'_{m-1})$ in time $\softO(m)$. These formulas
require that $m < p$ and that $S'$ be invertible modulo $S$, which is
ensured by our assumption that $S$ is separable. To compute ${\bf t}$
and ${\bf t'}$, we can use Shoup's power projection
algorithm~\cite{Shoup94}, which takes $O(m^{(\omega+1)/2})$ operations
in $k$.

Once $\zeta_0$ is known, to apply the mapping ${\varphi^\star_S}^{-1}$
to an element $g(x,t)$, we proceed coefficient-wise in $t$. Write
$g=\sum_{0 \le i < p} g_i(x) t^i$, with all $g_i$ of degree less than
$m$. Then
${\varphi^\star_S}^{-1}(g) 
= \sum_{0 \le i < p} \left (  g_i(\zeta_0) \bmod T \right)(x^p)\, x^i$
where $T$ is the polynomial obtained by raising all coefficients 
of $S$ to the power $p$, so that $S(x)^p=T(x^p)$.

Computing $T$ takes $O(m \log(p))$ operations in $k$; then, computing
each term $ g_i(\zeta_0) \bmod T$ can be done using the Brent-Kung
modular composition algorithm for $O(m^{(\omega+1)/2})$ operations in
$k$; the total is $O(m^{(\omega+1)/2} p)$. Finally, the evaluation at
$x^p$ and the summation needed to obtain ${\varphi^\star_S}^{-1}(g)$
do not involve any arithmetic operations.

\begin{rem}
\label{rem:cyclo}
In the case where $S = x^m - c$ (where $c \in k$ and $p$ does not 
divide $m$), there actually exists a quite simple explicit formula for 
${\varphi_S^\star}^{-1}$: it takes $t$ to $x$ and $x$ to $c^q x^{pn}$ 
where $n$ and $q$ are integers satisfying the B\'ezout's relation $pn + 
qm = 1$. Using this, one can compute ${\varphi_S^\star}^{-1}(g)$
in $\softO(pm)$ operations in $k$ in this special case.
\end{rem}

\noindent
{\bf Conclusion.}
Let us call {\tt phiS} and {\tt phiS\_inverse} the two subroutines
described above for computing $\varphi_S$ and its inverse respectively.
Proposition \ref{prop:psiSAp2} leads to the following algorithm for
computing the $p$-curvature modulo $S^p$.

\noindent\hrulefill

\noindent {\bf Algorithm} {\tt local\_p\_curvature}

\noindent{\bf Input:} a polynomial $S$ and a matrix $A_S \in M_r(k[x]/S^p)$

\noindent{\bf Output:} the $p$-curvature of the system $Y' = A_S \:Y$

\smallskip\noindent 1.\ $A_{S,\ell} = \texttt{phiS}(A_S)$

\noindent \phantom{1.\ }%
{\sc Cost:} $\softO(p r^2 m)$ operations in $k$ (with $m = \deg S$)

\smallskip\noindent 2.\ compute a fund. system of solutions
$Y_S \in M_r(\ell[t]/t^p)$

\noindent \phantom{2.\ }%
of the system $Y' = A_{S,\ell} Y$ at precision $p$.

\noindent \phantom{2.\ }%
{\sc Cost:} $\softO(p r^\omega)$ op. in $\ell$ using {\tt fundamental\_solutions}

\noindent \phantom{2.\ }%
{\sc Remark:} Here $\ell = k[x]/S$

\smallskip\noindent 3.\ $A_{p,\ell} = 
Y_S \cdot \coeff(A Y_S, \,p{-}1) \cdot Y_S^{-1}$

\noindent \phantom{3.\ }%
at precision $O(t^p)$

\noindent \phantom{3.\ }%
{\sc Cost:} $\softO(p r^\omega)$ operations in $\ell$

\smallskip\noindent 4.\ $A_p = \texttt{phiS\_inverse}(A_{p,\ell})$

\noindent \phantom{4.\ }%
{\sc Cost:} $\softO(p r^2 m^\omega)$ operations in $k$ in general

\noindent \phantom{4.\ }%
\phantom{{\sc Cost:} }$\softO(p r^2 m^{(\omega+1)/2})$ operations in $k$
if $m<p$

\smallskip\noindent 5.\ {\bf return} $A_p$.

\noindent\hrulefill

\smallskip

To conclude with, it is worth remarking that implementing the
algorithm {\tt local\_p\_curvature} can be done using usual power
series arithmetic: indeed, we only need to perform computations in the
quotient $\Sdp/\Sdp_{\geq p}$ which is isomorphic to $\ell[t]/t^p$ by
Corollary~\ref{cor:quoSdp}. Furthermore, we note that if we are using the
algorithm \texttt{fundamental\_solutions} at line 2, then $Y_S^{-1}$
can be computed by performing an extra loop in
\texttt{fundamental\_solutions}; indeed the matrix $Z$ we obtain this
way is exactly $Y_S^{-1}$.


\subsection{Gluing}
\label{subsec:glueing}

We recall that we have started with a differential system $Y' = AY$
(with $A = \frac 1{f_A} \tilde A$) and that our goal is to compute the
matrix $A_p$ of its $p$-curvature.  Lemma~\ref{lem:bound} gives bounds
on the size of the entries of $A_p$. We need another lemma, which
ensures that we can find enough small ``evaluation points'' (lying in
a finite extension of $k$). Let $\Fp$ denote the prime subfield of
$k$.

\begin{lem}
\label{lem:existSi}
Given a positive integer $D$ and a nonzero polynomial $f \in k[x]$, 
there exist pairwise coprime polynomials $S_1, \ldots, S_n \in \Fp[x]$ 
with $n \leq D$ such that:

\noindent
$\bullet$ $\sum_{i=1}^n \deg S_i \geq D$

\noindent
$\bullet$
for all $i$, the polynomial $S_i$ is coprime with $f$ and has
degree at most $1 + \log_p(D + \deg f)$.
\end{lem}

\begin{proof}
Let $m$ be the smallest integer such that $p^m \geq D + \deg f$.
Clearly $m \leq 1 + \log_q(D + \deg f) \leq 1 + \log_p(D + \deg f)$.
Let $\Fpm$ be an extension of $\Fp$ of degree $m$ and $K$ be the
compositum of $k$ and $\Fpm$.
Let $S_1, \ldots, S_t$ be the minimal polynomials over $\Fp$ (without 
repetition) of all elements in $\Fpm \subset K$ which are not a root of
$f$. We then have $\deg S_i \leq m$ for all $i$ and $\sum_{i=1}^t \deg 
S_i \geq p^m - \deg f \geq D$.  It remains now to define $n$ as the 
smallest integer such that $\sum_{i=1}^n \deg S_i \geq D$. Minimality 
implies $\sum_{i=1}^{n-1} \deg S_i < D$ and thus $n \leq D$.  Therefore 
$S_1, \ldots, S_n$ satisfy all the requirements of the lemma.
\end{proof}

The above proof yields a concrete algorithm for producing a sequence 
$S_1, \ldots, S_n$ satisfying the properties of Lemma \ref{lem:existSi}: 
we run over elements in $\Fpm$ and, for each new element, append its 
minimal polynomial over $\Fp$ to the sequence $(S_i)$ unless it is not 
coprime with $f$. We continue this process until the condition 
$\sum_{i=1}^n \deg S_i \geq D$ holds. Keeping in mind the logarithmic 
bound on $m$, we find that the complexity of this algorithm is at most 
$\softO(D + \deg f)$ operations in $k$. 
Let us call \texttt{generate\_points} the resulting routine: it takes as 
input the parameters $f$ and $D$ and return an admissible sequence $S_1, 
\ldots, S_n$.

We are now ready to present our algorithm for computing the 
$p$-curvature:

\noindent\hrulefill

\noindent {\bf Algorithm} {\tt p\_curvature}

\noindent{\bf Input:} a matrix $A$ written as 
$A = \frac 1{f_A} \cdot \tilde A$

\noindent{\bf Output:} the $p$-curvature of the differential system
$Y' = AY$

\smallskip\noindent 1.\ $S_1, \ldots, S_n = 
\texttt{generate\_points}(f_A,d+1)$

\noindent \phantom{1.\ }%
{\sc Cost:} $\softO(d)$ operations in $k$

\noindent \phantom{1.\ }%
{\sc Remark:} we have $n = O(d)$ and $\deg S_i = O(\log d)$, $\forall i$

\smallskip\noindent 2.\ {\bf for} $i=1,\dots,n$:

\noindent \phantom{2.\ }\hspace{0.4cm}$A_{i,p} = 
\text{\tt local\_p\_curvature}(S_i, A \bmod S_i^p)$

\noindent \phantom{2.\ }%
{\sc Cost:} $\softO(p d r^\omega)$ operations in $k$

\smallskip\noindent 3.\ compute
$B \in M_r(k[x])$ with entries of degree $\leq pd$

\noindent \phantom{3.\ }%
such that $B \equiv f_A^p \cdot B_i \pmod{S_i^p}$ for all $i$

\noindent \phantom{3.\ }%
{\sc Cost:} $\softO(pdr^2)$ operations in $k$

\smallskip\noindent 4.\ {\bf return} $\frac 1 {f_A^p} \cdot B$

\noindent\hrulefill

\smallskip

\noindent
In view of the previous discussion and Lemma \ref{lem:bound}, the 
correctness and the cost analysis of the algorithm \texttt{p\_curvature} 
are both straightforward. Hence, Theorem \ref{th:pcurv} is proved.

We conclude this subsection with three remarks.
First, when applying Chinese Remainder Theorem (CRT) on line 3 of 
Algorithm \texttt{p\_curvature}, we notice that all moduli $S_i^p$ are 
polynomials in $x^p$. This allows the following optimization. Writing
$f_A^p \cdot B_i \equiv \sum_{j=0}^{p-1} B_{i,j}(x^p) x^j \pmod {S_i^p(x)}$
and denoting by $C_j$ the unique solution of degree at most $d$ to the 
congruence system:
$$B_j(x) \equiv B_{i,j}(x) \pmod{T_i(x)}
\quad \text{where } T_i(x^p) = S_i^p(x)$$
we have $B = \sum_{j=0}^{p-1} B_j(x) x^j$. This basically allows us
to replace one CRT with polynomials of degree $dp$ by $p$ CRT with
polynomials of degree $d$. We save this way the polynomial factors 
in $\log(p)$ in the complexity.

Second, instead of working with $n$ polynomials $S_i$, one may 
alternatively choose a unique polynomial $S$ of the form $S = X^m - a$ 
where $m \geq d$ is an integer not divisible by $p$ and $a \in k$ are 
such that $S$ and $f_A$ are coprime. This avoids the use of Chinese 
Remainder Theorem and the resulting complexity stays in 
$\softO(pdr^\omega)$ provided that we use Remark \ref{rem:cyclo} in 
order to compute the inverse of $\varphi_S$.

Third, we observe that the algorithm {\tt p\_curvature} is very easily 
parallelizable. Indeed, each iteration of the main loop (on line 2) is 
completely independent from the others. Thus, they all can be 
performed in parallel. Moreover, according to the first remark (just 
above), the application of the Chinese Remainder Theorem (on line 3) 
splits into $p r^2$ smaller independent problems and can therefore be 
efficiently parallelized as well.


\begin{figure*}\centering
\renewcommand{\arraystretch}{1.2}
{\scriptsize 

\begin{tabular}{l|ll||>{\raggedleft}p{4em}|>{\raggedleft}p{4em}|>{\raggedleft}p{4em}|>{\raggedleft}p{4em}|>{\raggedleft}p{4em}|>{\raggedleft}p{4em}|>{\raggedleft}p{4em}|>{\raggedleft}p{4em}|l}
\cline{4-11}
 \multicolumn{3}{c|}{} & \multicolumn{8}{c|}{$\mathbf{p}$} \cr
\cline{4-11}
 \multicolumn{3}{c|}{} & \centering $\mathbf{157}$ & \centering $\mathbf{281}$ & \centering $\mathbf{521}$ & \centering $\mathbf{983}$ & \centering $\mathbf{1\:811}$ & \centering $\mathbf{3\:433}$ & \centering $\mathbf{6\:421}$ & \centering $\mathbf{12\:007}$ & \cr
\cline{2-11}
 & \multirow{2}{*}{$d = 5$,} & \multirow{2}{*}{$r = 5$} & $0.39$~s & $0.71$~s & $1.22$~s & $2.34$~s & $4.41$~s & $8.93$~s & $18.0$~s & $36.1$~s \cr
 &  &  & $0.26$~s & $0.76$~s & $2.69$~s & $9.05$~s & $32.6$~s & $145$~s & $593$~s & $2\:132$~s \cr
\cline{2-11}
 & \multirow{2}{*}{$d = 5$,} & \multirow{2}{*}{$r = 11$} & $1.09$~s & $2.05$~s & $3.65$~s & $7.05$~s & $12.6$~s & $26.7$~s & $53.3$~s & $109$~s \cr
 &  &  & $1.25$~s & $3.70$~s & $12.8$~s & $45.5$~s & $163$~s & $725$~s & $2\:942$~s & \hfill$-$\hfill\null \cr
\cline{2-11}
 & \multirow{2}{*}{$d = 5$,} & \multirow{2}{*}{$r = 20$} & $2.93$~s & $5.25$~s & $9.52$~s & $17.7$~s & $32.5$~s & $68.1$~s & $139$~s & $288$~s \cr
 &  &  & $4.29$~s & $12.4$~s & $42.5$~s & $153$~s & $548$~s & $2\:460$~s & \hfill$-$\hfill\null & \hfill$-$\hfill\null \cr
\cline{2-11}
 & \multirow{2}{*}{$d = 11$,} & \multirow{2}{*}{$r = 20$} & $6.89$~s & $13.3$~s & $22.6$~s & $45.0$~s & $80.4$~s & $167$~s & $342$~s & $711$~s \cr
 &  &  & $11.6$~s & $34.7$~s & $121$~s & $486$~s & $1\:943$~s & \hfill$-$\hfill\null & \hfill$-$\hfill\null & \hfill$-$\hfill\null \cr
\cline{2-11}
 & \multirow{2}{*}{$d = 20$,} & \multirow{2}{*}{$r = 20$} & $14.0$~s & $25.1$~s & $49.9$~s & $94.0$~s & $176$~s & $357$~s & $733$~s & $1\:472$~s \cr
 &  &  & $27.0$~s & $84.5$~s & $314$~s & $1\:283$~s & \hfill$-$\hfill\null & \hfill$-$\hfill\null & \hfill$-$\hfill\null & \hfill$-$\hfill\null \cr
\cline{2-11}
\end{tabular}
}
\renewcommand{\arraystretch}{1}

\smallskip

\centering
{\scriptsize Running times obtained with Magma V2.19-4 on an AMD Opteron 
6272 machine at 2GHz and 8GB RAM, running Linux.\hspace{0.5em}\null}

\caption{Average running time on random inputs of various sizes}\label{fig:table}
\end{figure*}


\subsection{The case of differential operators}

To conclude with, we would like to discuss the case of a differential
operator
$L = a_r \partial^r + a_{r-1} \partial^{r-1} + \cdots +
a_1 \partial + a_0$
with $a_i \in k[x]$ for all $i$, of maximal degree $d$. 

Recall that the $p$-curvature of $L$ is that of the differential
module $(\ring\langle \partial\rangle /\ring\langle
\partial\rangle L, \partial_{-C})$, where $C$ is the companion matrix
associated to $L$ as in~\eqref{eq:comp}. Applying directly the
formulas in Proposition \ref{prop:psiSAp2} requires the knowledge of
the solutions of the system $Y' = -C Y$.
It is in fact easier to compute solutions for the system $X'=\trp{C}
X$, since we saw that these solutions are the vectors of the form
$\trp{(y, y', \ldots, y^{(r-1)})}$, where $y$ is a solution of $L$.
This is however harmless: the $p$-curvatures $A_p$ and $B_p$ of the
respective systems $Y'=-C Y$ and $X'=\trp{C} X$ (which are so-called
adjoint) satisfy $A_p = -\trp{B_p}$. Thus, we can use the formulas
given above to compute $\varphi_S(B_p)$, and deduce $\varphi_S(A_p)$
for a negligible cost. Equivalently, one may notice that the
fundamental matrices of solutions of our two systems are transpose of
one another, up to sign.

Moreover, instead of using the second formula of Proposition 
\ref{prop:psiSAp2} to compute the local $p$-curvatures, we recommend 
using the first one, which is
$\varphi_S(B_p) = - X_S \cdot X_S^{(p)}(0) \cdot X_S^{-1}$ where
$X_S$ is a fundamental system of solutions of $X' = \trp{C}X$ and
$\bar X_S$ denotes its reduction in $M_r(\ell[t]/t^p)$. If $f_0,
\ldots, f_{r-1}$ are solutions of the system \eqref{eq:systop}, the
$(i,j)$-th entry of $X_S$ is just $f_j^{(i)}$. Hence the matrices
$\bar X_S$ and $X_S^{(p)}(0)$ can be obtained from the knowledge of
the image of $f_i$'s modulo $\Sdp_{\geq p+r}$ just by reorganizing
coefficients (and possibly multiplying by some factorials depending on
the representation of elements of $\Sdp$ we are using). 

As for the $f_i$'s, they can be computed by the algorithm
\texttt{solutions\_operator} (provided its assumptions are
satisfied). We need finally to compute $X_S^{-1}$: since $X_S(0)$ is
the identity matrix, this can be done either using Newton iterator, a
divide-and-conquer approach or a combination of both, which computes
the inverse of $X_S$ at a small precision, and uses divide-and-conquer
techniques for higher ones (the latter being the most efficient in
practice).  All these remarks do speed up the execution of our
algorithms when $d$ is not too large compared to $r$.

Last but not least, we notice that, in the case of differential 
operators, the matrix $A_p$ is easily deduced from its first column. 
Indeed, writing $A_p = (a_{i,j})_{0 \leq i,j < r}$ and letting
$c_j = a_{r-1,j} \partial^{r-1} + \cdots + a_{1,j} \partial + 
a_{0,j} \in k(x)\langle\partial\rangle$
be the differential operator obtained from the $j$-column of $A_p$,
it is easily checked that $c_{j+1}$ is the remainder in the Euclidean
division of $\partial c_j$ by $L$. Comparing orders, we further find
$c_{j+1} = \partial c_j - \frac{\lc(c_j)}{a_r} L$
where $\lc(c_j)$ is the leading coefficient of $c_j$. This remark is 
interesting because it permits to save memory: indeed, instead of storing 
all local $p$-curvatures $A_{p,\ell}$, we can just store their first 
column. Doing this, we can reconstruct the first column of $A_p$ using 
the Chinese Remainder Theorem (\emph{cf} \S \ref{subsec:glueing}) and 
then compute the whole matrix $A_p$ using the recurrence.


\section{Implementation and timings}\label{sec:timings}


We implemented our algorithms in Magma in the case of differential 
operators; the source code is available at 
\url{https://github.com/schost}. 
Figure~\ref{fig:table} gives running times for random operators of 
degrees $(d,r)$ in $k[x]\langle \partial \rangle$ and compares them with 
running times of (a fraction free version of) Katz's algorithm which 
consists in computing the recursive sequence $(A_i)$ until $i=p$. In 
each cell, the first line (resp. the second line) corresponds to the 
running time obtained with our algorithm (resp. Katz's algorithm); a
dash indicates that the corresponding running time exceeded one hour. 
Our benchmarks rather well reflect the predicted dependence with respect 
to $p$: quasi-linear for our algorithm and quadratic for Katz's 
algorithm.

Larger examples (than those presented in Fig.~\ref{fig:table}) are also 
reachable: for instance, we computed the first column of the 
$p$-curvature of a ``small'' multiple of the operator $\phi_H^{(5)}$ 
considered in~\cite[Appendix~B.3]{BoHaMaZe07} modulo the prime $27449$. 
This operator has bidegree $(d,r) = (108,28)$. The computation took 
about 19 hours and the size of the output in human-readable format is 
about 1GB (after bzip2 compression, it decreases to about 300MB).

{\tiny

}

\end{document}